\documentclass[pra,twocolumn,superscriptaddress,showpacs]{revtex4-1}

\usepackage{amsmath}
\usepackage{latexsym}
\usepackage{amssymb}
\usepackage{graphicx}
\usepackage[colorlinks=true, citecolor=blue, urlcolor=blue]{hyperref}
\usepackage{float}
\usepackage{amsfonts}
\usepackage{textcomp}
\usepackage{mathpazo}
\usepackage{enumitem}
\usepackage{comment}

\newcommand{\ket}[1]{| #1 \rangle}
\newcommand{\bra}[1]{\langle #1 |}

\usepackage{bbm}

\usepackage{hyperref}

\usepackage[draft]{fixme}
\usepackage{amsmath,bbm}
\usepackage{graphicx}
\usepackage{amsfonts}
\usepackage{amssymb}
\usepackage{amsmath, amssymb, amsthm,verbatim,graphicx,bbm}
\usepackage{mathrsfs}
\usepackage{color,xcolor,longtable}
\usepackage[rightcaption]{sidecap}

\newcommand{\beq}[0]{\begin{equation}}
\newcommand{\eeq}[0]{\end{equation}}

\def\ra{\rangle}
\def\la{\langle}

\def\be{\begin{equation}}
\def\ee{\end{equation}}
\def\ben{\begin{eqnarray}}
\def\een{\end{eqnarray}}
\def\eea{\end{array}}
\def\bea{\begin{array}}

\newcommand{\Tr}[1]{\mathrm{Tr}#1}
\newcommand{\bei}{\begin{itemize}}
\newcommand{\eei}{\end{itemize}}
\def\ra{\rangle}
\def\la{\langle}

\newcommand{\I}{\mathbbm{1}}

\renewcommand{\emph}[1]{\textbf{#1}}

\makeatletter
\newtheorem*{rep@theorem}{\rep@title}
\newcommand{\newreptheorem}[2]{%
\newenvironment{rep#1}[1]{%
 \def\rep@title{#2 \ref{##1}}%
 \begin{rep@theorem}}%
 {\end{rep@theorem}}}
\makeatother

\theoremstyle{plain}
\newtheorem{theorem}{Result}
\newreptheorem{theorem}{Result}
\newtheorem*{thm*}{Result}

\newtheorem{defn}{Definition}

\theoremstyle{definition}

\theoremstyle{remark}

\usepackage[T1]{fontenc}

\begin{document}
\title{Demonstration of quantum correlations that are incompatible with absoluteness of measurement }
\author{Shubhayan Sarkar}
\affiliation{Center for Theoretical Physics, Polish Academy of Sciences, Aleja Lotnikow 32/46, 02-668 Warsaw, Poland }
\author{Debashis Saha}
\affiliation{S. N. Bose National Centre for Basic Sciences,
Block JD, Sector III, Salt Lake, Kolkata 700098, India }
\affiliation{School of Physics, Indian Institute of Science Education and Research Thiruvananthapuram, Kerala 695551, India }

\begin{abstract}
Exploiting the tension between the two dynamics of quantum theory (QT) in the Wigner's Friend thought experiment, we point out that the standard QT leads to inconsistency in observed probabilities of measurement outcomes between two super-observers - Wigner and his Student. To avoid such inconsistent predictions of QT, we hypothesize two distinct perspectives.
  The first one is "Absoluteness of measurement (AoM)," that is, any measurement process is an absolute event irrespective of other observers and yields a single outcome. The other is "Non-absoluteness of measurement (NoM)" as the negation of AoM.    We introduce an operational approach, first with one friend and then with two spatially separated friends, to test the validity of these two perceptions in quantum theory without assuming the details of the experiment. First, we show that the set of probabilities obtainable for NoM is strictly larger than the set obtainable for AoM.  We provide the simplest scenario so far, involving a single quantum preparation and one unitary operation by a super-observer that can demonstrate correlations incompatible with AoM.  Remarkably, in the scenario with spatially separated observers, we present a strict hierarchy among the sets of probabilities observed in the following three theories: classical or local realist, quantum theory with AoM, and quantum theory with NoM.  
\end{abstract}

\maketitle
\section{Introduction}
Quantum theory (QT) has been one of the most successful theories describing nature. Even with the counter-intuitive predictions, QT has always been validated by experiments till date, not just for microscopic systems, but even large macroscopic states \cite{mole2,mole3,mole4,mole5,mole6}, complex molecules \cite{mole7, molecule} and living systems \cite{vir}. Despite this, a long sought-after question in physics is whether there exists some scale above which QT is not valid. From a theoretical perspective, the postulates of QT are valid for any system. As a consequence, QT, in general, can be regarded to be universally valid unless we find some experiment that refutes the predictions of QT at some scale.

The universal validity of QT has far-reaching consequences, as was first observed by Schrodinger in 1935 \cite{Schrod}. Famously known as Schrodinger's cat paradox, the thought experiment involves putting a cat into a state of superposition of two different macroscopic states, dead or alive. The existence of such a state, even if classically incomprehensible, does not contradict quantum laws of physics. However, such a scenario leads to some discrepancies in quantum theory, which is known as the "measurement problem" of QT. 
The issue was differently illustrated by Wigner in 1967 \cite{Wigner}  using a thought experiment involving two different observers who give two different descriptions of the same physical process of measuring a superposed quantum state. The first observer, named Friend, measures a quantum system, and the second observer, named Wigner, describes Friend along with the experiment he performs, that is, Friend's laboratory. The discrepancy inherently lies in the fact that QT allows Wigner to describe Friend's measurement as a reversible or unitary process, while according to Friend, the measurement changes the quantum state irreversibly. However, 
such a discrepancy can be avoided if we do not associate quantum state with any underlying reality of physical systems and take the standpoint that quantum state only represents some knowledge about physical systems. This thought experiment is commonly referred to as the Wigner's Friend (WF) scenario and is rigorously dealt with later.

One can comprehend that the discrepancy arises in the WF scenario because of the presence of two different dynamics within QT, one when systems are evolving with time and the other when systems are measured. Even when such a discrepancy, or the so-called measurement problem, itself is counter-intuitive, it does not give rise to any disagreement at the empirical level. Moreover, the "measurement problem" is a problem for the orthodox interpretation of QT, wherein the quantum state is considered to be the realist description of a physical system. This leads to a significant question of whether the tension between the two dynamics of QT results in any observable inconsistency. In other words, what would be the operational version of the inconsistency arising in the WF thought experiment? If any inconsistency persists at the empirical level, it would be particularly interesting to find different solutions that can resolve those inconsistencies and probe whether those solutions can be distinguished operationally.

In this work, we first revisit the Wigner's Friend scenario, which gives rise to inconsistency at the empirical level due to the two different dynamics of QT. To resolve this inconsistency, we propose two hypotheses, namely "Absoluteness of Measurement (AoM)," inspired by the assumption of "Absoluteness of observed events" \cite{Wiseman} and the negation of it termed here as "Non-Absoluteness of Measurement (NoM)." To test the validity of these two hypotheses operationally, we generalize the WF thought experiment with the Friend being able to choose different measurements based on different inputs along with Wigner being able to perform different unitaries on Friend's laboratory.  We show that the set of probabilities obtainable for AoM is a proper subset of the probabilities obtained for NoM. In order to do so, we find the simplest scenario involving a single quantum preparation and one unitary operation by a super-observer in which quantum predictions with NoM (or quantum predictions in unitary quantum theory) are incompatible with AoM. Compared to the previous works \cite{Renner, Brukner, Wiseman}, our scenario is minimal as only two observers are required to observe a violation. Further on, we do not require the spatially separated observers or entanglement to detect a violation of AoM, making our scheme more relevant from a practical perspective. 

To detect the difference between QT with AoM and NoM, we propose inequalities satisfied by the quantum predictions with AoM, and their violation persists even if one accounts for non-idealness in the Wigner's Friend setup. In fact, there is a considerable gap between the bound obtained using AoM and NoM, which can be observed in experiments. Remarkably, our generalization for bipartite systems presents a strict hierarchy among the correlation sets observed in the following three theories: classical or local realist, quantum with AoM, and quantum with NoM or unitary quantum. This brings forth deeper insight into the quantum correlations in the Wigner's Friend setup. 

\section{Revisiting Wigner's Friend}\label{sec2}

We first revisit the WF experiment wherein the tension between two different dynamics of QT leads to inconsistency at the empirical level. Subsequently, we propose two perspectives for a consistent description when applying QT.

\subsection{Wigner's Friend and Student thought experiment}\label{sec:wfs}

The modified Wigner's Friend and Student (WFS) scenario \cite{Wigner} can be described as follows: Friend, who is confined in an isolated laboratory, receives a quantum system $(Q_s)$, performs a measurement on it, and subsequently, observes a definite outcome. Let us term the isolated laboratory, including Friend, measurement device, and the environment inside it, as "Lab". Wigner and his Student, on the other hand, describe the physical process of the combined system of Lab and $Q_s$. A logical inconsistency regarding the state of the combined system arises if QT is supposed to be universal, as defined below.
\begin{defn}
[Universal Quantum Theory (UQT)] The standard quantum theory   (prescribed by the following postulates)    is applicable to all physical systems, including macroscopic systems like an observer
\footnote{Note that, here we only refer to the operational theory and we do not presume any existence of realism or ontology.}. 
  
\begin{itemize}
    \item The state space of any isolated physical system is a Hilbert space $\mathcal{H}$. The system is represented by a state vector, that is, a unit vector $|\psi\ra \in \mathcal{H}$.
    \item The time evolution of the closed system is described by unitary transformation. That is, the state of the system in the later time is given by $U|\psi\ra$, where $U$ is a unitary operator acting on $\mathcal{H}$.
    \item Any $d \ (1\leqslant d \leqslant \infty)$ outcome measurement is represented by a set of $d$ projectors $\{P_i\}_{i=1}^d$ such that $\sum_i P_i = \I$, where $\I$ is the identity operator on $\mathcal{H}$. The probability of obtaining outcome $i$ on the physical system (represented by $
    |\psi\ra$) is given by $\Tr(P_i|\psi\ra\!\la\psi|)$.
    If outcome $i$ is obtained after the measurement, then the state of the system changes to $P_i|\psi\ra/\sqrt{\Tr(P_i|\psi\ra\!\la\psi|)}$.
    \item The state space of a composite physical system is given by the tensor product of the respective state spaces of the component physical systems.
\end{itemize}
   
\end{defn}
To reflect on the above fact, let us take an explicit example (see Figure \ref{fig}). The observer Friend receives a quantum state $\ket{\psi}\in\mathbbm{C}^2$ and obtains an outcome after performing Pauli-Z measurement $(\sigma_z)$ on it. The physical interaction between Lab and $Q_s$ is certainly a measurement process with respect to Friend.   However, the postulates of universal QT do not specify which dynamic applies for describing the interaction between Lab and $Q_s$ with respect to a super-observer (Wigner or Student). QT allows Wigner and Student to describe the same physical interaction either by time-evolution dynamics according to the Schrodinger equation or by post-measurement dynamics according to the collapse postulate. 
Say that according to Wigner the combined state of Lab and $Q_s$ evolves unitarily. Now, we know that if the initial state of $Q_s$ is an eigenvector of $\sigma_z$, that is, $\ket{0}$ or $\ket{1}$, then after the interaction Friend observes a definite outcome $+1$ or $-1$ and the state $Q_s$ remains unchanged. This empirical fact implies that the unitary $U$ describing the evolution of the combined system must satisfy the following conditions,
\beq \label{eq1}
 U\ket{0}\ket{f} = \ket{0}\ket{f_+}, \ U\ket{1}\ket{f} = \ket{1}\ket{f_-},
\eeq 
where $\ket{f}$ is the initial state of Lab, and $\ket{f_\pm}$ is the final state of the Lab such that the measurement device shows a definite outcome $\pm1$ and Friend observes a definite outcome $\pm1$. Besides, whenever Friend performs a measurement, she encodes the information that she has obtained a definite outcome on an ancillary system $\ket{\psi}_{an}\in\mathbbm{C}^2$ (this could be a classical system) and keeps outside his/her Lab. Subsequently, the complete descriptions of Lab, $Q_s$, and ancilla with respect to Wigner are
\ben \label{wevo}
 U\ket{0}\ket{f} \otimes \sigma_x \ket{0}_{an} = \ket{0}\ket{f_+} \otimes\ket{1}_{an}, \nonumber  \\ U\ket{1}\ket{f} \otimes \sigma_x \ket{0}_{an} = \ket{1}\ket{f_-} \otimes \ket{1}_{an},
\een 
wherein Friend performs the  unitary operation $\sigma_x$ to encode the information that he/she performed a measurement and observed an outcome.  
Notice that the ancilla does not contain any information about the outcome of the measurement.   
On the other hand, Student describes the evolution of the combined system of Lab and $Q_s$ according to the collapse postulate of UQT as 
\beq 
 \ket{0}\ket{f} \rightarrow \ket{0}\ket{f_+}, \ \ket{1}\ket{f} \rightarrow \ket{1}\ket{f_-},
\eeq 
whereas the ancillary system evolves as
\beq \label{anevo}
\sigma_x \ket{0}_{an} = \ket{1}_{an}.
\eeq

Therefore, if the initial state of $Q_s$ is $(\ket{0}+\ket{1})/\sqrt{2}$, then it follows from \eqref{wevo} that the final state of the combined system of Lab and $Q_s$, according to Wigner is given by,
\beq \label{FSW}
 U\left(\frac{\ket{0}+\ket{1}}{\sqrt{2}}\right)\ket{f}
=  \frac{1}{\sqrt{2}}\left( \ket{0}\ket{f_+} + \ket{1}\ket{f_-} \right) .
\eeq  
While Student,  who does not have knowledge about the outcome of the measurement, describes the following evolution of the combined system of Lab and $Q_s$,  
\begin{eqnarray} \label{FSF}
\left(\frac{\ket{0}+\ket{1}}{\sqrt{2}}\right)&\ket{f}&\rightarrow \nonumber\\ 
\frac{1}{2}\big( &\ket{0}\!\bra{0}&\otimes\ket{f_+}\!\bra{f_+} + \ket{1}\!\bra{1}\otimes\ket{f_-}\!\bra{f_-} \big).
\end{eqnarray}
Additionally, since the ancilla is not correlated with the outcome of the measurement, it evolves the same way as given in \eqref{anevo}. The purpose of the ancillary system is for Friend to communicate outside the Lab that the measurement has been performed and a definite outcome has been observed. Moreover, if Wigner applies an operation on Lab to alter Friend's observation, then, later on, by measuring the ancilla Friend can be sure that she indeed obtained an outcome before. 
To summarize, we note that the two observers Wigner and Student ascribe two different descriptions, given by \eqref{FSW} and \eqref{FSF}, for the same physical system comprising Lab and $Q_s$.

\subsection{Operational Consistency}\label{sec:consistency}

Firstly, we note that according to Wigner the state of the combined system of Lab and $Q_s$ is given by \eqref{FSW}, while Student describes the combined system as given in \eqref{FSF}. As a consequence, the quantum predictions of observed probabilities of measurements on the combined system are inconsistent. For instance, if the following observable is measured
\beq \label{Fm}
\{\ket{F}\!\bra{F}, \I - \ket{F}\!\bra{F} \}
\eeq 
where $\ket{F}=1/\sqrt{2}\left( \ket{0}\ket{f_+} + \ket{1}\ket{f_-} \right)$, the first outcome will be certain as far as Wigner is concerned. On the contrary, Student anticipates that both outcomes are equally likely.
The discrepancy at the observable level between Wigner and Student is unavoidable in UQT as it does not prescribe which of the two descriptions is correct. 

  
To resolve this inconsistency, we consider an additional assumption on the `measurement process' so that one can infer a unique state using quantum theory, that is, either \eqref{FSW} or \eqref{FSF}, without any ambiguity. One such assumption can be the absolute nature of an observed measurement event. 
%
\begin{defn}[Absoluteness of measurement (AoM)]\label{AOM}
Any measurement performed by some observer yields a single outcome, and that measurement is an absolute event irrespective of other observers and processes.
\end{defn}
This assumption is in the same spirit as considered in Ref. \cite{Wiseman} under the name `Absoluteness of observed events'. However, `Absoluteness of observed events' is theory-independent, while in this work, we will consider AoM along with the validity of the standard quantum theory.  AoM implies that if any observer (e.g., Friend here) performs a measurement on a quantum system, it qualifies as a measurement for every other observer. As a consequence, only the post-measurement dynamic is applicable to describe the system for all observers whenever any one observer observes the measurement outcome. Therefore, quantum theory with AoM unambiguously infers the combined state to be \eqref{FSF}, which  is the description by Student. 
In contrast to AoM, we can suppose the negation of AoM in order to retain consistency. 
%

\begin{defn} [Non-absoluteness of measurement (NoM)]\label{AOE}
The necessary condition for applying post-measurement dynamics on a system by an observer is to observe the measurement outcome. In other words, measurement is not an absolute event with respect to an observer unless that observer observes an outcome.
\end{defn}
That is, any observer without the knowledge of the measurement outcome describes any physical process by the time evolution specified by the theory.  Unitary quantum theory, wherein the collapse dynamics is a derived fact instead of a postulate, meets NoM. Unitary quantum theory has been studied in the context of Wigner's Friend experiment \cite{commphy,praArmando}. However, there may exist other interpretations of quantum mechanics that obey NoM. For instance, one may consider NoM along with all the standard postulates, including the collapse postulate, so that altogether it does not exhibit any inconsistency.

In the context of the WFS experiment, NoM simply implies that the measurement on the quantum system performed by Friend is not a measurement process for super-observers unless they
have the knowledge about the outcome of the measurement. Therefore, the correct description of the combined system with respect to any super-observer is \eqref{FSW},  even if Friend describes the quantum system using post-measurement dynamics. Clearly, Wigner represents the perspective of NoM. \\
\begin{widetext}
\onecolumngrid
\begin{figure}[h!]
    \begin{center}
    \includegraphics[width=0.85\textwidth]{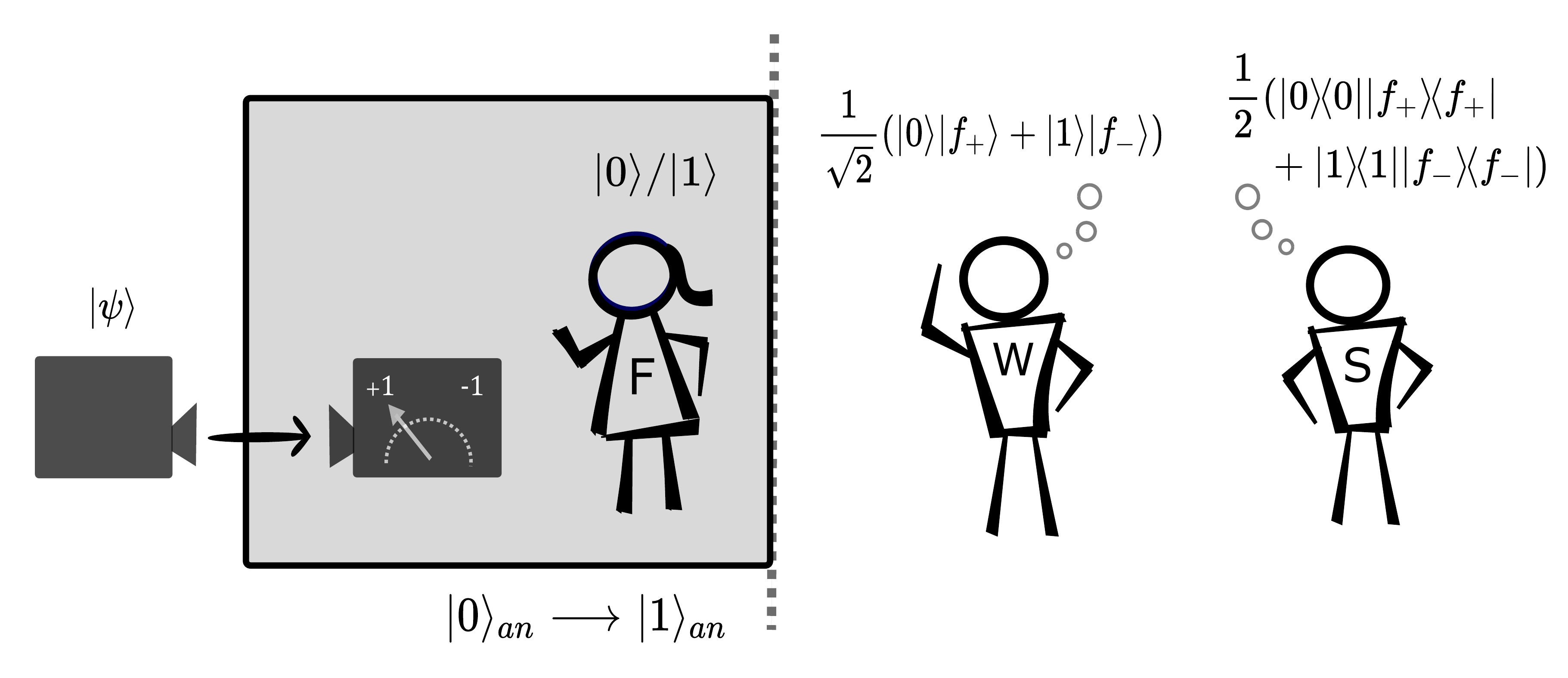} 
    \caption{\textbf{Simplest Wigner's Friend and Student (WFS) scenario.} Friend (F) receives a system from the preparation device on which she performs a measurement. Friend encodes the information about the fact that measurement has been performed by reversing the ancillary bit $\ket{0}_{an}$ to $\ket{1}_{an}$ which can be read by Wigner (W), Student (S) and Friend at any further times. Now, Wigner (or Student) can apply different processes on the Lab and finally perform a measurement on the Lab to test their distinct perspectives.} 
    \label{fig}
       \end{center}
\end{figure}
\end{widetext}



For any physical theory, in order to avoid an observable discrepancy among these two perspectives, we have two options - either all the observable probabilities should not depend upon whether a measurement is considered to be absolute or not; or the status of a measurement, whether it is absolute or not, is specified by the theory itself. Strictly speaking, any reasonable universal physical theory should satisfy the following notion of consistency. 
\begin{defn}[Operationally Consistent theory for all Observers (OCO)]\label{def:oc}
  A universal theory is `operationally consistent for all observers' if any two observers, after applying the theory correctly, always predict unique probabilities that can be observed by both  of them.
\end{defn}
Any theory with unique dynamics, like classical theory, is always operationally consistent. Within QT, as the quantum state of the combined system of Lab and $Q_s$ is different depending on the perspective, we can conclude that the UQT is operationally inconsistent unless it specifies either AoM or NoM for every measurement. 
 %

\section{An operational approach to witness correlations that are incompatible with Absoluteness of measurements}

We address here the following question: Is it possible to empirically test AoM def-\eqref{AOM} or NoM def-\eqref{AOE} without the knowledge of the underlying details of the experiment? In other words, can we discriminate ‘UQT with AoM’ and ‘UQT with NoM’ in an operational way?   So far, we have taken into consideration two super-observers (Wigner and Student) in order to understand the distinct predictions within quantum theory; however, in what follows, we regard one super-observer as sufficient to carry out the empirical tests introduced below. Secondly, for simplicity, we assume that the super-observer is restricted to applying unitary transformations to the combined system before making any measurements. However, the results obtained in the following sections also hold for transformations that can be an arbitrary unital quantum channel. 
Let us first describe the Wigner's Friend setup in full generality. Every run of the experiment consists of the following steps:  
%
%
  
\begin{itemize}[leftmargin=\parindent]
    \item At $t_0$, Friend and the super-observer receive some random variable $x\in \{0,\dots,n-1\}$ and $w\in \{0,\dots,m-1\}$, respectively. So, in general $x$ takes $n$ possible values, while $w$ takes $m$ possible values. The initial state of Lab knowing the input $x$ is $\ket{f^x}$.
    After reading $x$, Friend keeps this random variable outside the Lab. 
    
    \item At $t_1$, a preparation device prepares a quantum system ($Q_s$) in the state $\ket{\psi}$ which enters the isolated Lab of Friend.
    
    \item At $t_2$, depending on $x$, Friend performs a measurement, denoted by $\overline{A}_x$, on $Q_s$ and obtains an outcome. At this point, Friend encodes the information of observing a definite outcome on the ancilla as mentioned before. The ancilla does not contain information about the outcome as seen by Friend. We consider the measurement $\overline{A}_x$ always to be a $d$-outcome projective measurement defined by the following rank-one projectors
\beq \label{barA}
\overline{A}_x := \left\{\overline{P}^x_i \right\}^{d-1}_{i=0} . 
\eeq   
    Here, we identify Friend's measurement with an over-line. Note that Wigner and Student do not know the value of $x$. But, they know the possible values of $x$ along with the complete description of $\ket{f^x}$ given every $x$. We denote $\ket{f^x_i}$ for the state of the Lab such that the measurement device shows outcome $i$ and Friend observes outcome $i$ after the measurement \eqref{barA}. Note that, $\la f^x_i|f^{x'}_j\ra = \delta_{i,j}\delta_{x,x'}$.
      
    For any rank-one projective measurement $\overline{P}^x_i=|\psi^x_i\ra\!\la\psi^x_i|$ for every $i,x$ and the state $|\psi\ra$ can be expressed as
\beq \label{psixi}
|\psi\ra = \sum_{i=0}^{d-1} \alpha^x_i |\psi^x_i\ra ,
\eeq 
wherein $\forall x, \sum_i |\alpha^x_i|^2 =1$. It follows from the previous discussion that after the measurement is performed, the state of the combined system according to AoM and NoM, respectively, are
\be \label{labsaom}
\sum^{d-1}_{i=0} |\alpha^x_i|^2 \ket{F^x_i}\!\bra{F^x_i} 
\ee 
and 
\be \label{labsnom}
\sum_{i=0}^{d-1} \alpha^x_i |F^x_i\ra ,
\ee 
where 
\beq \label{F}
    |F^x_i\ra = |\psi^x_i\ra |f^x_i\ra .
\eeq

    \item At $t_3$, the super-observer performs a unitary transformation $U_{w}$ on the combined system depending on $w$. Throughout this article, we further assume that for $w=0$, the super-observer does not do anything. The unitaries $U_w$ should be in the form \eqref{Uw}, which will be discussed soon.
    
    \item At $t_4$, the super-observer performs the following $(dn+1)$-outcome measurement on the combined system,
    \beq  \label{A}
    \Omega :=  \left\{\{P^{x}_i\}^{d-1,n-1}_{i=0,x=0} \ , \I - \sum_{i,x}  P^x_i \right\}
    \eeq   
    where
    \beq \label{P}
    P^x_i = \overline{P}^x_i \otimes |f^x_i\ra\!\la f^x_i| 
    \eeq 
    is a projector acting on the combined state of Lab and $Q_s$, $|f^x_i\ra$ is the state of Lab observing the outcome $i$ after measuring $\overline{A}_x$, and $\overline{P}^x_i$ is given in \eqref{barA}. 
    Note that measurement \eqref{A} is in the same basis as Friend's measurement on $Q_S$, along with measuring the input $x$. 
    The dimension of the macroscopic system Lab is presumably larger than $d\cdot n$, and thus, an auxiliary outcome is taken into account to fulfill the completeness condition. However, the probability of observing this outcome in the following thought-experiment is zero.
   The super-observer and Friend should agree on the value of input $x$, which was kept outside the Lab. For that reason, an additional condition, that the $d$-dimensional subspace spanned by $\{P^x_i\}^{d-1}_{i=0}$ remains invariant under the application of unitary for each $x$, is imposed on $U_w$.   Suppose the unitary is such that $U_w|F^x_i\ra = \sum_i \beta_i |F^x_i\ra + \gamma \ket{F^{x'}_j} $ for some non-zero value of $\gamma$, then after the measurement \eqref{A} there will be a non-zero probability that Friend is in a state of observing input $x'$ that is different from the actual input $x$. Thus, the unitaries are in the form 
   \be \label{Uw}
   U_w = \bigoplus_x U^x_w, 
   \ee 
   where $U_w^x$ is acting on the subspace of $P^x_i$.
       
    Apart from the value of $x$, note that the super-observer and Friend observe the same outcome, say, $a\in \{0,\dots,d-1\}$.  

    \item At $t_5$, Friend gets the value of $w$, matches the value of $x$ that was kept outside the Lab, and measures the ancilla to ensure that she obtained a definite outcome before. Moreover, by measuring 
    $Q_s$ in the basis \eqref{barA}, Friend can verify whether Wigner's measurement was \eqref{A} or not, since the state of $Q_s$ should collapse to $\ket{\psi^x_a}$ for outcome $a$. Thus, if the $d$-dimensional subspace spanned by $\{P^x_i\}^{d-1}_{i=0}$ is not invariant after applying $U_w$ for some $x$, then Friend will detect a discrepancy. 
    
\end{itemize}
This is repeated many times so that the super-observer obtains the empirical probability defined as
\begin{center}
\begin{tabular}{c c c c c c}
 & $t_0$ & $t_1$ & $t_2$ & $t_3$ & $t_4$ \\ 
 & $\downarrow$ & $\downarrow$ & $\downarrow$ & $\downarrow$ & $\downarrow ~ $ \\  
$p(a|\overline{A}_x,U_w) :=$ $p(a\ |$ & $x,w,$ & $\psi,$ & $\overline{A}_x,$ & $U_{w},$ & $\Omega)$,
\end{tabular}
\end{center} 
in which all the conditional variables are mentioned on the right-hand side according to the sequence of time. 


In what follows, $\{p(a|\overline{A}_x,U_w)\}$ denote the set of empirical probabilities 
without assuming the form of the state $|\psi\ra$, measurements $\overline{A}_x$, and the unitaries $U_w$. Depending on the hypotheses AoM and NoM, we can obtain two sets of empirical probabilities. The first set is given by,
\begin{defn}\label{qcm}
[Quantum Correlations with absoluteness of measurement (QCAoM)]
Set of all observed probabilities $\{p(a|\overline{A}_x,U_w)\}$ obtained in the setup, employing universal QT with the additional condition that measurements are absolute according to definition \eqref{AOM}.
\end{defn}
and the second set is given by,
\begin{defn}\label{qce}
[Quantum Correlations with non-absoluteness of measurement (QCNoM)]
Set of all observed probabilities $\{p(a|\overline{A}_x,U_w)\}$ obtained in the setup, employing universal QT with the additional condition that measurements are not absolute according to definition \eqref{AOE}.
\end{defn}


Let us now propose a simple test in this extended Wigner's Friend setup that shows the existence of correlations in "UQT with NoM (def-\eqref{AOE})" but not in "UQT with AoM (def-\eqref{AOM})".

\subsection{An elementary test for correlations that cannot be observed in "UQT with AoM"}\label{sub1}
Consider the scenario where the state $|\psi\rangle \in \mathbbm{C}^2$ is arbitrary and the measurement performed by Friend in an arbitrary direction whose eigenstates are $|\psi_+\ra, |\psi_{-}\ra$.  
According to quantum theory along with Absoluteness of Measurement \eqref{AOM} the combined state of $Q_s$ and Lab after the measurement by Friend is obtained similar to \eqref{FSF},
\be \label{labsAoM}
p_+ \ket{F_+}\!\bra{F_+} + p_- \ket{F_-}\!\bra{F_-} ,
\ee 
where $p_{\pm} = |\la \psi|\psi_{\pm}\ra|^2$ is the probability of getting an outcome $\pm1$ with respect to Friend, and 
\be 
\ket{F_\pm} = |\psi_{\pm}\ra|f_{\pm}\ra,
\ee  
where $\ket{f_\pm}$ is the final state of the Lab such that the measurement device shows a definite outcome $\pm1$ and Friend observes a definite outcome $\pm1$. Following a similar argument to obtaining \eqref{FSW}, the quantum description of the combined state of $Q_s$ and Lab with Non-absoluteness of Measurements is 
\be \label{labsNoM}
\la \psi_+|\psi\ra \  \ket{F_+} + \la \psi_-|\psi\ra \ \ket{F_-} .
\ee 
We consider the scenario where the super-observer sometimes applies
a reversible process $U$ on the combined system of Lab and $Q_s$. Finally, the super-observer opens the Lab and asks Friend to reveal the outcome, that is, performs the following measurement 
\be 
\{|F_+\ra\!\la F_+|,|F_-\ra\!\la F_-|, \I - |F_+\ra\!\la F_+| - |F_-\ra\!\la F_- | \} 
\ee 
for the respective outcome denoted by $\{+1,-1,\emptyset\}$.
The last element of the above measurement is due to the fact that the dimension of the Lab and $Q_s$ has to be at least four but in general, may be much higher. With full generality, the Hilbert space of the Lab is larger than $\mathcal{H}_{L'}=\text{Span}\{\ket{f_+},\ket{f_-}\}$. When the combined state of the Lab and $Q_s$ does not belong $\mathbbm{C}^2\otimes\mathcal{H}_{L'}$, then the last element of Wigner's measurement would show up. One can also understand this outcome as taking into account the null results of Wigner's measurement. 

Let us now evaluate the probability of obtaining outcome $\pm1$ of this measurement  for AoM after applying $U$ using \eqref{labsAoM},
\begin{eqnarray}\label{ppmU}
p(\pm1|U)&=&\sum_{a=+,-}\Tr\left(p_aU \ket{F_a}\!\bra{F_a}U^\dagger \ket{F_\pm}\!\bra{F_\pm}\right)\nonumber\\
&\leqslant & \max\{p_+,p_-\} \times \nonumber \\
&& \quad \Tr \bigg(U \underbrace{(\ket{F_+}\!\bra{F_+} + \ket{F_-}\!\bra{F_-})}_{\leqslant \I}U^\dagger \ket{F_\pm}\!\bra{F_\pm}\bigg) \nonumber \\
&\leqslant & \max\{p_+,p_-\} .
\end{eqnarray} 
On the other hand, the probability of obtaining outcome $\pm1$ for AoM when the super-observer does not apply any transformation is just
\begin{eqnarray}\label{ppmI}
p(\pm1|\I) &=& \Tr\left((p_+ \ket{F_+}\!\bra{F_+} + p_- \ket{F_-}\!\bra{F_-}) \ket{F_\pm}\bra{F_\pm}\right)  \nonumber \\
&=& p_\pm . 
\end{eqnarray} 
Let us now consider the expression
\be \label{T}
T = p(+1|U) - \left|p(+1|\I) - 1/2 \right| - \left|p(-1|\I) - 1/2 \right| ,
\ee 
which is a simple function of observed probabilities.
Replacing the probabilities in $T$ from \eqref{ppmU}-\eqref{ppmI} and using $p_++p_-=1$, we find that 
\begin{eqnarray} \label{Tleq1/2}
T &=& \max\{p_+,1-p_+\} - 2 \left|p_+- 1/2 \right| \nonumber \\
&\leqslant & 1/2 ,
\end{eqnarray} 
for AoM.
Remarkably, this relation holds \textit{irrespective} of the initial quantum state $|\psi\ra,$ the measurement $\{\ket{\psi_+},\ket{\psi_-}\}$ by Friend and the unitary $U$ by the super-observer.
However, the value of $T$ \eqref{T} can be even 1 in "UQT with NoM". Consider the quantum state
\beq 
\ket{\psi}=1/\sqrt{2}(\ket{\psi_+}+\ket{\psi_-}),
\eeq 
so that the combined state of $Q_s$ and Lab is
$1/\sqrt{2}(  \ket{F_+} + \ket{F_-} ) $ according to \eqref{labsNoM}. Clearly, for this state $p(\pm1|\I) = 1/2$. 
Furthermore, let the unitary applied by the super-observer be
\begin{eqnarray}
U = \frac{1}{\sqrt{2}}
\begin{pmatrix}
1&1\\-1&1
\end{pmatrix}\nonumber
\end{eqnarray}
written in the two-dimensional basis, $\{\ket{F_+},\ket{F_-}\}$. One can readily verify that the state after applying this unitary becomes $\ket{F_+}.$ Therefore, $p(+1|U)=1$, which in turn implies $T=1$. 

The consequence of this above observation is remarkable. If any super-observer observes a violation of condition $T\leqslant 1/2$, then he or she must be convinced, without knowing the details of the experiment, that the quantum theory with AoM can not reproduce all the empirical observations. Importantly, such a test of AoM is experimentally robust in the sense that the upper bound on $T$ will be higher than 1/2 when experimental non-idealities are taken into account. Consequently, if one observes the value of $T$ close to 1, then also it is possible to refute the hypothesis quantum theory with AoM. 

Moreover, the relation $T\leqslant 1/2$ holds true even in "UQT with AoM" if the super-observer applies an arbitrary \textit{unital} quantum channel on the combined system of $Q_s$ and Lab instead of applying perfect unitary. A general unital quantum channel $\Lambda$ on an operator $A$ is defined by as $\Lambda(A) = \sum_i K_iAK_i^\dagger$ such that $\sum_i K_iK^\dagger_i = \I $, where $\{K_i\}$ are the Kraus operators satisfying $\sum_i K^\dagger_iK_i =\I$.  In the case of a unital quantum channel applied by super-observer the observed probability in quantum theory with AoM can be obtained similar to \eqref{ppmU} as,
\begin{eqnarray}\label{ppmL}
p(\pm1|\Lambda)&\leqslant & \max\{p_+,p_-\} \times \nonumber \\
&& \Tr \bigg( \underbrace{\sum_i K_i(\ket{F_+}\!\bra{F_+} + \ket{F_-}\!\bra{F_-})K_i^\dagger }_{\leqslant \I} \ket{F_\pm}\!\bra{F_\pm}\bigg) \nonumber \\
&\leqslant & \max\{p_+,p_-\} .
\end{eqnarray} 
Here, we have used the unital nature $\sum_i K_iK^\dagger_i = \I $ as follows
\begin{eqnarray}
\I &=& \sum_i K_i (\ket{F_+}\!\bra{F_+} + \ket{F_-}\!\bra{F_-} ) K_i^\dagger \nonumber \\ 
&& + \sum_i K_i (\I - \ket{F_+}\!\bra{F_+} - \ket{F_-}\!\bra{F_-} ) K_i^\dagger \nonumber \\
& \geqslant & \sum_i K_i (\ket{F_+}\!\bra{F_+} + \ket{F_-}\!\bra{F_-} ) K_i^\dagger. 
\end{eqnarray}
To get from the second to the third line of the above expression, we used the fact that $K_iPK_i^\dagger$ are positive semi-definite  for any projector $P$. As a result,  \eqref{Tleq1/2} holds true even when Wigner can apply a unital quantum channel.


Now, we show that every correlation achievable in universal QT assuming absoluteness of measurements (QCAoM def-\eqref{qcm}) is a proper subset of correlations obtainable in universal QT assuming Non-absoluteness of measurement (QCNoM def-\eqref{qce}).

\subsection{QCAoM is a subset of QCNoM}
We would first show that QCAoM def-\eqref{qcm} is a subset of QCNoM def-\eqref{qce}.

\begin{theorem}\label{thm:subset}
QCAoM is a subset of QCNoM, that is, QCAoM $\subseteq$ QCNoM.
\end{theorem}
\begin{proof}
We show that any probability $p(a|\overline{A}_x,U_w)$ that is obtained from arbitrary initial state $|\psi\ra$, measurement $\overline{A}_x$ and unitary $U_w$, can also be obtained within QCNoM. Taking the general form of $\overline{A}_x$ given in \eqref{barA} and using \eqref{labsaom} for AoM, the state  after $t_3$ is given by
\beq 
\rho_L=\sum^{d-1}_{i=0} |\alpha^x_i|^2 U_w  \ket{F^x_i}\!\bra{F^x_i} (U_w)^\dagger 
\eeq 
where $\ket{F^x_i}$ is defined in \eqref{F}.
As a consequence, the probability of obtaining outcome $a$ when Friend performs the measurement $\overline{A}_x$ is
\ben  
p(a|\overline{A}_x,U_w) &=& \sum\limits^{d-1}_{i=0} |\alpha^x_i|^2 |\la F^x_a|U_w|F^x_i\ra |^2 \nonumber\\
&=& \begin{cases} 
|\alpha_a^x|^2, & \text{if } w=0 \\
\sum\limits^{d-1}_{i=0} |\alpha^x_i \la F^x_a|U_w^x|F^x_i\ra |^2, & \text{if } w\neq 0
\end{cases}
\label{pof}
\een 
where $U_w=\bigoplus_x U_w^x$ and $U_w^x$ is acting on the $d$-dimensional subspace spanned by $\{|F^x_i\ra\}^{d-1}_{i=0}$. 
Let us consider a strategy in QCNoM in which the initial state and the measurements are the same as before. For $w=0$, the state after $t_3$ is given by \eqref{labsnom}.
For $w=0$, the probability of getting outcome $a$ is $|\alpha_a^x|^2$, which is the same as \eqref{pof}, while for $w\neq 0$, the applied unitary by Wigner is $\bigoplus_x \Tilde{U}^x_w $ so that
\beq \label{Uxt}
\Tilde{U}_w^x \left(\sum_{i=0}^{d-1} \alpha^x_i |F^x_i\ra\right) = \sum_{a=0}^{d-1} \sqrt{\beta_{a,w,x}} |F_a^x\ra 
\eeq 
where 
\beq 
\beta_{a,w,x}  = \sum^{d-1}_{i=0} |\alpha^x_i \la F^x_a|U^x_w|F^x_i\ra |^2 .
\eeq 
Finally, note that the transformed state at $t_3$ in the above Eq. \eqref{Uxt} is such that the probability $p(a|\overline{A}_x,U_w)$ is exactly the same as obtained for QCAoM in Eq. $\eqref{pof}$. This completes the proof.
\end{proof} 

\subsection{Correlations that are incompatible with QCAoM for rank-one measurement by Friend}

Before providing explicit examples of operational witnesses of the difference between QCAoM and QCNoM, we first obtain a relation that holds true for QCAoM.

\begin{theorem}\label{theo1}
In QCAoM, for any $d$-outcome rank-one projective measurement $\overline{A}$ given in \eqref{barA}, the following holds true for all $a=0,\dots ,d-1,$
\begin{eqnarray}\label{eq:thm1}
p(a|\overline{A},U)\leqslant \max_{i} p(i|\overline{A},\I) ,
\end{eqnarray}
where $U$ is any reversible process that can be applied by a super-observer on the combined system of $Q_s$ and Lab. Note that, we omit the random variable $x$ in the probability since this relation holds for every $x$.
\end{theorem}

\begin{proof}
According to universal quantum theory along with Absoluteness of Measurement \eqref{AOM} the combined state of $Q_s$ and Lab after time $t_3$ is
\begin{eqnarray}\label{stateboth}
\rho_L= \sum^{d-1}_{i=0} p_i \ket{F_i}\!\bra{F_i}  ,
\end{eqnarray}
where $p_{i}$ is the probability of getting  outcome $i$ with respect to Friend (see Eq. \ref{labsaom}), and $\ket{F_i}$ is defined in \eqref{F}.   Since the measurement $\overline{A}$ is considered to be rank-one, it is defined by the set of orthonormal states $\{|\psi_i\ra\}$ given by \eqref{psixi}. Consequently, in the final measurement by the super-observer \eqref{A},
\be 
P_i =\ket{F_i}\!\bra{F_i}
\ee 
are also rank-one projectors.     Thus, for any reversible process $U$, 
the left-hand side of \eqref{eq:thm1} is
\begin{eqnarray}\label{proboutcome}
p(a|\overline{A},U)&=&\Tr\left(U\rho_LU^\dagger\ket{F_a}\!\bra{F_a}\right)\nonumber\\
&=& \sum_{i=0}^{d-1} \ p_i \ \underbrace{\Tr \left( U \ket{F_i}\!\bra{F_i} U^\dagger \ket{F_a}\!\bra{F_a} \right)}_{q_i} \nonumber \\
&=& \sum_{i=0}^{d-1} \ p_i \ q_i 
\end{eqnarray} 
for all $a = 0,\dots,d-1$. In the second line of the above equation, we denote the expression with trace by $q_i$ for each $i$. On the other hand, using the fact that $\sum_i \ket{F_i}\!\bra{F_i} \leqslant \I$ we know  
\be 
\sum_i U\ket{F_i}\!\bra{F_i} U^\dagger \leqslant \I .
\ee 
Since $\ket{F_a}\!\bra{F_a}$ is a rank-one projector, it follows from the above relation that
\be \label{qi}
\sum_{i=0}^{d-1} q_i \leqslant 1. \ee 
Due to fact that $p_i,q_i$ are all non-negative and satisfy the relations $\sum_i p_i =1$ as well as Eq. \eqref{qi}, the right-hand side of Eq. \eqref{proboutcome} is upper bounded by the maximum value of $p_i$, that is, 
\be p(a|\overline{A},U) \leqslant \max_{i=0,\dots,d-1} p_i \ .\ee
Finally, by noting that $p_i=p(i|\overline{A},\I)$, we obtain \eqref{eq:thm1}.
\end{proof} 
Let us now explore correlations that are not compatible with QCAoM but exist in QCNoM. Employing Result \ref{theo1} and the idea presented in Sec. \ref{sub1}, we propose a simple expression to distinguish between QCAoM and QCNoM. Consider the extended Wigner's Friend scenario described above, where Friend does not receive any variable $x$, and $w \in \{0,1\}$. Given any measurement $\overline{A}$ of the form \eqref{barA} having $d$ outcomes, Wigner seeks to maximize the following empirical expression over all possible preparations $|\psi\ra$ and reversible transformations $U$:  
\be \label{Tq}
T(\Vec{q}) =  p(0|\overline{A},U) - \sum_{i=0}^{d-1}|p(i|\overline{A},\I) -q_i|  
\ee
which is a function of the set of $d$ variables $\Vec{q}=(q_0,\dots,q_{d-1})$ such that $q_i \in [0,1)$ and $\sum_i q_i=1$.

\begin{theorem}\label{thm:Tq}
Within QCAoM where $\overline{A}$ is any $d$-outcome rank-one projective measurement, the maximum value of $T(\Vec{q})$ in Eq.\eqref{Tq} is bounded as,
\begin{equation}\label{Tqb}
  T(\Vec{q}) \leqslant \max\{q_0,\dots,q_{d-1}\}.  
\end{equation}
Moreover, there exists a realization in QCNoM such that $T(\Vec{q})=1$ for any values of $\Vec{q}$.
\end{theorem}
\begin{proof}
Let us first consider that the maximum among the probabilities $\{p(i|\overline{A},\I)\}_i$ which is denoted by
\begin{eqnarray}
p(i^*|\overline{A},\I)=\max_i\{p(i|\overline{A},\I)\}.
\end{eqnarray}
Now, we can always express this quantity as
\be\label{41}
p(i^*|\overline{A},\I) = \max\{q_0,\dots,q_{d-1}\} + t,
\ee
for some $t$ such that $0\leq p(i^*|\overline{A},\I)\leq 1$. Using Result \ref{theo1} for expressing $p(0|\overline{A},U)$, and the normalization of probabilities, we see that the following holds within QCAoM
\ben
T(\Vec{q}) &\leqslant &  \max \{q_0,\dots,q_{d-1}\} + t - |p(i^*|\overline{A},\I) - q_{i^*}| \nonumber \\
&&- \sum_{i\neq i^*} \left|p(i|\overline{A},\I) - q_i \right|.
\een
Now, using the identity that $|A|+|B|\geq |A+B|$, we have that
\ben
T(\Vec{q}) &\leqslant &  \max \{q_0,\dots,q_{d-1}\} + t - |p(i^*|\overline{A},\I) - q_{i^*}| \nonumber \\
&& - \left|\sum_{i\neq i^*} p(i|\overline{A},\I) - \sum_{i\neq i^*} q_i \right| \nonumber \\
& = & \max \{q_0,\dots,q_{d-1}\} + t - 2|p(i^*|\overline{A},\I) - q_{i^*}|, \nonumber \\
\een 
where to arrive at the third line of the above expression, we used the fact that $\sum_{i} p(i|\overline{A},\I)=\sum_{i}q_i=1$.
From the above relation, it trivially follows that for $t<0$, Eq. \eqref{Tqb} holds. Now if $t> 0$, then from Eq. \eqref{41} $p(i^*|\overline{A},\I) - q_{i^*}\geq p(i^*|\overline{A},\I) - \max_i\{q_{i}\}= t$.
Therefore, the maximum value of the right-hand-side of the above equation is obtained for $t=0$, and hence, \eqref{Tqb} is true. The upper bound in QCAoM is achieved when $Q_s$ is prepared in the state,
\be \label{mis}
 \sum_i q_i \ket{\psi_i}\!\bra{\psi_i},
\ee
or,
\be \label{is}
 \sum_i \sqrt{q_i}\ket{\psi_i}
\ee
where $|\psi_i\ra$ are eigenvectors of $\overline{A}$. Let us denote $q_{i_m} = \max\{q_0,\dots,q_{d-1}\}$, and choose any unitary such that $U|F_{i_m}\ra = |F_0\ra$. The combined state \eqref{labsaom} after the measurement is $\sum_i q_i \ket{F_i}\!\bra{F_i},$ which results the second term of $T(\Vec{q})$ in \eqref{Tq} to be zero. For this choice of unitary, $p(0|\overline{A},U)=q_{i_m}$.
Note that the realization \eqref{mis} is classical, and therefore, the upper bound in \eqref{Tqb} also holds for classical theory. 

We now state a realization within universal QT assuming Non-absoluteness of Measurement def-\eqref{AOE} that can attain the value $T(\Vec{q})=1$ \eqref{Tq} for any $\Vec{q}$. The quantum state of $Q_s$ is given by \eqref{is}.
Universal QT with NoM assigns the state  $\sum_i \sqrt{q_i} \ket{F_i}$ to the combined system of $Q_s$ and Lab. Thus, the second term of $T(\Vec{q})$ in \eqref{Tq} is zero. Any unitary so that 
\be 
U \left(\sum_i \sqrt{q_i} \ket{F_i} \right) = \ket{F_0}
\ee 
yields $p(0|\overline{A},U)=1$.
\end{proof}
   

Violation of \eqref{Tqb} can also be understood as an operational witness of correlations that are incompatible with the absoluteness of measurements def-\eqref{AOM}. In the later part of this manuscript, we would discuss its implications  towards different interpretations of QT. Now, we extend the above-introduced though-experiment to two spatially separated parties, and show that there exist correlations in universal QT with NoM def-\eqref{AOE} but not in universal QT with AoM def-\eqref{AOM},   irrespective of the rank of Friend's measurement.   

\begin{widetext}

\begin{figure}[h!]
   \begin{center}
    \includegraphics[width=0.85\textwidth]{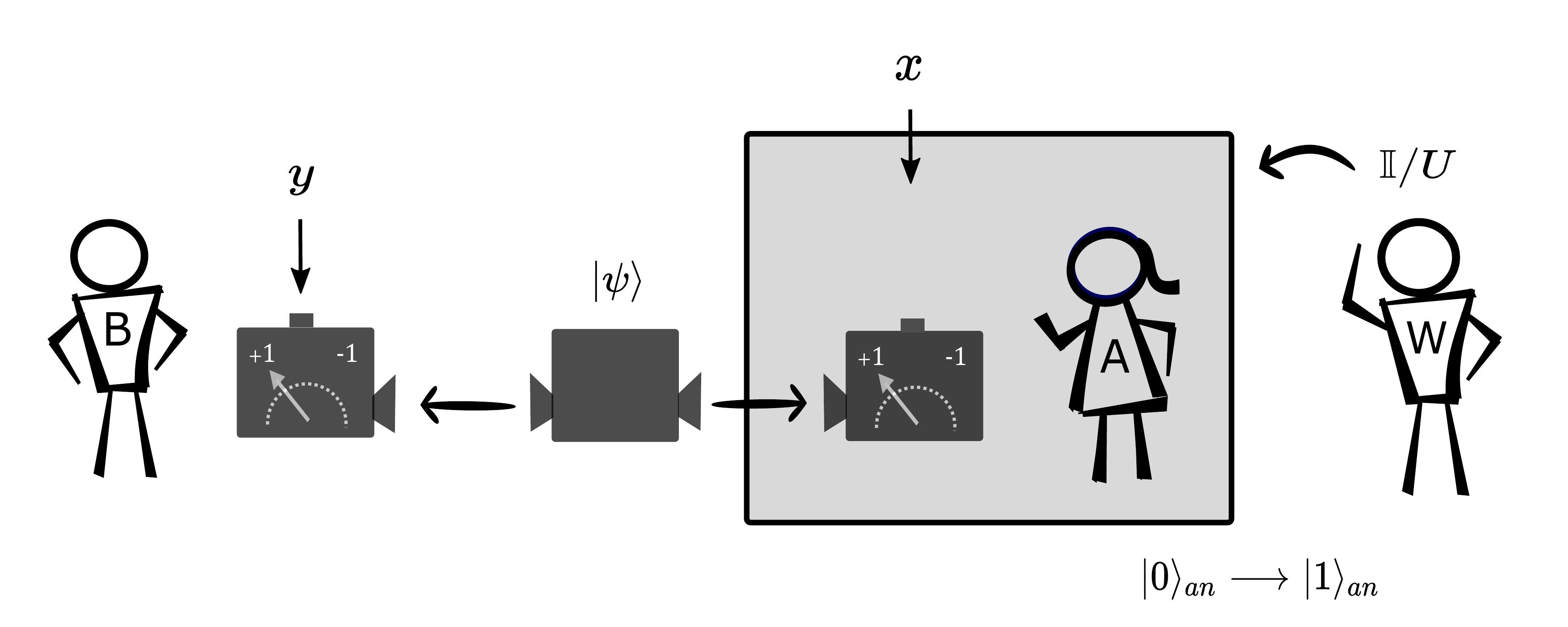} 
    \caption{\textbf{Activation of bipartite correlations.} Wigner's Friend, named Alice (A), and another observer, Bob (B), receive a system from the preparation device on which they perform measurements according to inputs $x,y,$ respectively, in their spatially separated Labs. Alice encodes the information about the fact that the measurement has been performed by reversing the ancillary bit $\ket{0}_{an}$ to $\ket{1}_{an}$, which can be read by Wigner, Student, and Alice at any further times. Alice and Bob are not allowed to communicate with each other. Wigner (W) or Student can apply a reversible process on Alice's Lab after her measurement is performed, and finally, Wigner or Student opens Alice's Lab to know the outcome of her measurement. By repeating this experiment many times, Wigner together with Bob observes the joint probabilities of the experiment. }
    \label{fig1}
       \end{center}
\end{figure}

\end{widetext}

\subsection{Bipartite correlations that are incompatible with QCAoM}

In general, we can extend the WF setup involving single quantum systems to a WF setup involving multipartite quantum systems. Let us consider the scenario illustrated in Fig. \ref{fig1}, wherein Wigner's Friend, named as Alice, and another spatially separated observer, say, Bob, perform a measurement on the subsystems of a bipartite quantum system. And the super-observer, Wigner or Student, can apply local reversible processes on Alice's Lab. Each run of this two-party version of the experiment comprises the same steps as before on Alice's side, along with the following additional steps on Bob's side.

\begin{itemize}[leftmargin=\parindent]

\item At $t_0$, apart from Alice and Wigner, Bob also receives an input variable that is denoted by $y$.

\item At $t_1$, the preparation device prepares a bipartite system with one part of the system entering Alice's Lab and the other part entering Bob's Lab.

\item At $t_2$, besides Alice's measurement $\overline{A_x}$ \eqref{barA}, Bob performs a measurement $B_y$ depending on $y$ that results an outcome $b$.

\item At $t_3$, Wigner applies some reversible operation denoted by $U_w$ on Alice's Lab, each of which keeps the subspace spanned by $\{|F^x_i\ra\}^{d-1}_{i=0}$ invariant. Wigner does not apply any operation for $w=0$.

\item At $t_4$, Wigner and Student perform the measurements \eqref{A} on Alice's Lab to obtain outcome $a$.
\end{itemize}  

After repeating many times, Wigner and Bob communicate with each other to obtain the joint probability distribution of outcomes $a,b$ for different inputs $x,y,w$. 
For convenience, let us drop the terms $\overline{A}_x, B_y,$ $\Omega$, and use the following concise notation for these joint probabilities:
\begin{center}
\begin{tabular}{c c c c c c}
 & $t_0$ & $t_1$ & $t_2$ & $t_3$ & $t_4$ \\ 
 & $\downarrow$ & $\downarrow$ & $\downarrow$ & $\downarrow$ & $\downarrow ~ $ \\  
$p(a,b|x,y,U_w) :=$ $p(a,b\ |$ & $x,y,w,$ & $\psi,$ & $\overline{A}_x,B_y,$ & $U_{w},$ & $\Omega)$ .
\end{tabular}
\end{center} 

Since the operations by Wigner and Bob are spatially separated, the obtained joint probabilities should satisfy the notion of causality or the no-signalling conditions. 
\begin{defn}[No-signalling conditions]\label{Nosigcond}
The local statistics of Alice as well as Wigner is independent of Bob's measurement settings and vice versa. 
\begin{eqnarray}
 \forall x,y,y',w, \quad   \sum_{b}p(a,b|x,y, U_w)=\sum_{b}p(b,a|x,y', U_w)\nonumber
\end{eqnarray}
    and 
\begin{eqnarray}
\forall x,x',y,w,w', \quad \sum_{a}p(a,b|x,y, U_w)=\sum_{a}p(b,a|x',y, U_{w'}) . \nonumber
\end{eqnarray} 
\end{defn}

As we show, there exists an empirical test that can single out QCNoM from QCAoM. The empirical test involves Alice, Bob, and Wigner as described above in which $x,y,a,b,w \in \{0,1\}$. Consider the following two quantities denoted by $P_w$ that corresponds to $w=0,1,$ respectively,
\begin{eqnarray}\label{SP1}
    P_0 = \sum_{a,b,x,y} c_{a,b,x,y} \ p(a,b|x,y,\I) , \\
    \label{SP2}
    P_1=\sum_{a,b,x,y} c_{a,b,x,y} \ p(a,b|x,y,U) ,
\end{eqnarray}
where 
\be 
c_{a,b,x,y} = 
\begin{cases}
\frac{1}{4}, \text{ if } a\oplus b = x.y \\
0, \text{ otherwise.}
\end{cases}
\ee 
Note that the expression \eqref{SP1} is nothing but the CHSH expression \cite{CHSH} written in the Wigner's Friend scenario.   It is well known that the maximum values of standard CHSH expression $P_0$ are 3/4 in classical theory and $1/2(1+1/\sqrt{2})\approx 0.854$ in quantum theory.    The quantity $P_1$ captures the presence of super-observer. Now, we are ready to state the main result whose explicit proof is provided in Appendix \ref{app}.

\begin{theorem}\label{thm:P1}
The following inequality holds true for QCAoM,
\begin{equation}
\label{bP1m}
P_1 \leqslant \min\left\{ \frac12 \left( 1+\frac{1}{\sqrt{2}} \right) , \max \left\{ P_0, \frac32 - P_0 \right\} \right\} .
\end{equation}
\end{theorem} 
The above result has interesting consequences. As shown in Figure \ref{fig:p0p1}, Eq. \eqref{bP1m} provides a bound on $P_1$ depending on the value of $P_0$ . This can be used to propose the following inequality, which is satisfied by QCAoM and violated in QCNoM. 
\begin{figure}[h!]
    \centering
    \includegraphics[width=0.35\textwidth]{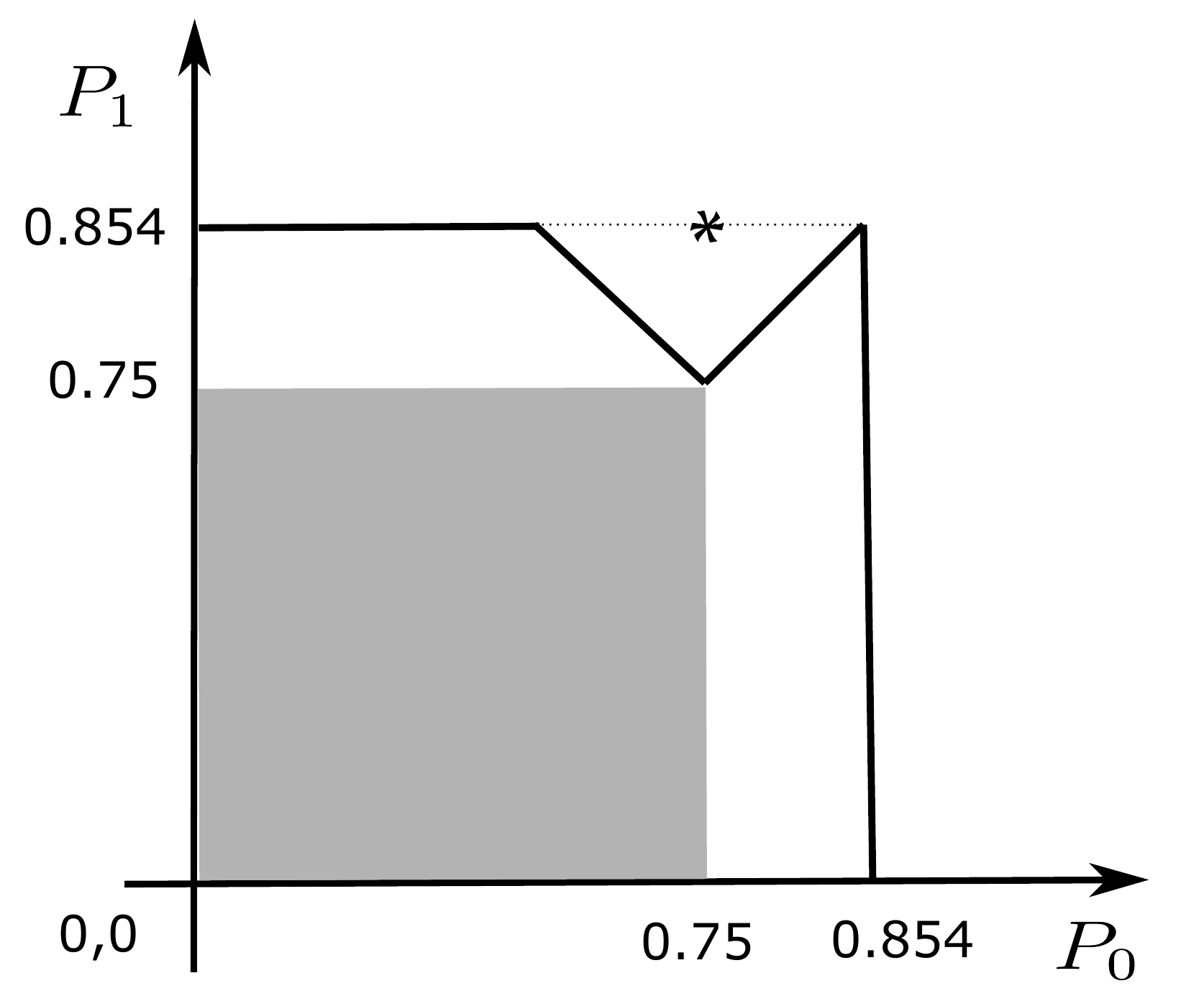}
    \caption{The values of $P_0,P_1$ that can be achieved from classical correlations are within the gray region, and those from QCAoM are within the boundary by the black line. While it is possible to obtain the values of $P_0,P_1$ marked by `*' in QCNoM as pointed out in Result \ref{thm:acg}. }
    \label{fig:p0p1}
\end{figure}

\begin{theorem}\label{thm:acg}
The following inequality holds true in QCAoM,
\begin{eqnarray}\label{bcP}
P_S = P_1 - \left|P_0- \frac{3}{4} \right| \leqslant \frac{3}{4} . 
\end{eqnarray}
Moreover, there exists a realization in QCNoM such that 
\beq \label{tbm}
P_S = \frac{1}{2}\left(1+\frac{1}{\sqrt{2}}\right). 
\eeq 
\end{theorem}
\begin{proof}
For a given quantum strategy, say, the value of $P_0=q$.  Substituting the upper bound of $P_1$ from \eqref{bP1m} we obtain the figure of merit \eqref{bcP} within QCAoM.
  Let us consider the first case, when $\max\{q, 3/2-q \}\leqslant 1/2(1+1/\sqrt{2}).$ From \eqref{bP1m}, we have that
  
\beq 
P_S \leqslant \max\{q, 3/2-q \} - |q - 3/4| = 3/4 
\eeq 
for any value of $q \in [0,1]$. 
  In the second case, where $\max\{q, 3/2-q \} \geqslant 1/2(1+1/\sqrt{2})$, the fact that $q\leqslant 1/2(1+1/\sqrt{2})$ implies that
\begin{eqnarray}
\frac{3}{2}-q\geqslant \frac{1}{2}\left(1+\frac{1}{\sqrt{2}}\right) ,
\end{eqnarray}
which means $q\leqslant \frac{3}{2}-\frac{1}{2}\left(1+\frac{1}{\sqrt{2}}\right)$. Using this bound of $q$ and condition \eqref{bP1m}, we compute $P_S$  to get
\be
P_S \leqslant \frac{1}{2}\left(1+\frac{1}{\sqrt{2}}\right) - \left|\frac{3}{2}-\frac{1}{2}\left(1+\frac{1}{\sqrt{2}}\right) - \frac34 \right| \leqslant \frac34.
\ee 
  
We now show that there exists QCNoM that achieves the value of $P$ given by \eqref{tbm}. For this, the maximally entangled state $\ket{\psi} = (1/\sqrt{2})\left(\ket{00}+\ket{11}\right)_{AB}$ is prepared on which Alice and Bob choose the measurements
\beq 
\overline{A}_0=\sigma_z, \overline{A}_1=\sigma_x, \ B_0=\sigma_z, B_1=\sigma_x . 
\eeq 
Now, the description of the joint system with respect to the super-observer assuming NoM def-\eqref{AOE} is determined by the unitary evolution after the measurements are done. Let us first look at the combined state of Alice's Lab and the system on Bob's side when the measurements are $\overline{A}_0=\sigma_z$ and $\overline{A}_1=\sigma_x$,
\begin{eqnarray}
& \ket{\Psi}_{0}=\frac{1}{\sqrt{2}}\left(\ket{0}\ket{f_{0}^0}\ket{0}+\ket{1}\ket{f_{1}^0}\ket{1}\right)_{LB}, \label{Psi0}\\
& \ket{\Psi}_{1}=\frac{1}{\sqrt{2}}\left(\ket{+}\ket{f_{0}^1}\ket{+}+\ket{-}\ket{f_{1}^1}\ket{-}\right)_{LB} . \label{Psi1}
\end{eqnarray}
It can be readily verified from the above expressions that,
\beq
\forall x,y=0,1, \ p(0,0|x,y,\I) = p(1,1|x,y,\I) = \frac{1}{2},
\eeq 
and thus, $P_0 = 3/4$.
For $w=1$, say, the super-observer performs the reversible operation $U=U^0\bigoplus U^1$, where
\begin{eqnarray}
U^x=
\begin{pmatrix}
\cos(\frac{\pi}{8})&-\sin(\frac{\pi}{8})\\\sin(\frac{\pi}{8})&\cos(\frac{\pi}{8})
\end{pmatrix}
\end{eqnarray}
written in the basis $\{\ket{0}\ket{f_{0}^x},\ket{1}\ket{f_{1}^x}\}$.
The joint probabilities after the unitary is applied should be evaluated on the combined state $U^x\ket{\Psi}_x$, wherein $\ket{\Psi}_x$ is given in \eqref{Psi0}-\eqref{Psi1}.
It follows from some simple calculations that the joint probabilities are
\begin{eqnarray}
p(0,0|0,0,U)+p(1,1|0,0,U)&=&\cos^2\left(\frac{\pi}{8}\right),\nonumber\\
p(0,0|0,1,U)+p(1,1|0,1,U)&=&\frac{1}{2}\left(1+\sin\left(\frac{\pi}{4}\right)\right),\nonumber\\
p(0,0|1,0,U)+p(1,1|1,0,U)&=&\frac{1}{2}\left(1+\sin\left(\frac{\pi}{4}\right)\right),\nonumber\\
p(0,0|1,1,U)+p(1,1|1,1,U)&=&\cos^2\left(\frac{\pi}{8}\right),
\end{eqnarray}
such that Eq.\eqref{tbm} holds. 
\end{proof}
 
Violation of Eq. \eqref{bcP} can be understood as an activation of quantum non-locality in the bipartite scenario when the description of the measurement process of a Friend is via "UQT with NoM". Notice that the strategy to violate the success probability $P_S$ \eqref{bcP} in  "UQT with NoM" is when the correlations shared among Alice and Bob, in the absence of Wigner, are local. When Wigner applies the unitary on Bob's Lab, then the local correlations are transformed into non-local ones, thus activating a resource that could not be possible when using "UQT with AoM". Consequently, we showed above that a resource can be generated when a Friend's Lab is described using "UQT with NoM" instead of "UQT with AoM". Further, the maximal value obtainable using "UQT with AoM" is the same when using classical or local strategies.



\section{Conclusions}

In recent times, there has been a renewed interest in the WF scenario due to the no-go theorem put forward by Frauchiger and Renner \cite{Renner}, where it is shown that QT is incompatible with three assumptions, namely universality, consistency, and single-outcome. 
Any interpretation of QT compatible with all three assumptions cannot be true. For showing the contradiction, a scenario consisting of four different observers was constructed where two observers behave like Friends and the other two behave like Wigner from the original Wigner's Friend scenario. As shown in Ref. \cite{Renner}, the observers predict different results for the same experiment if the assumptions of universality, consistency, and single-outcome hold. A similar no-go theorem exploiting Bell inequalities \cite{Bell, Bell66} was proposed by Brukner \cite{Brukner} imposing the assumption that measurement outcomes are objective fact. Contrary to the Frauchiger-Renner approach, the scenario involved two Friends who are spatially separated, with each receiving a qubit on which two local measurements are performed and two super-observers who have access to their Labs each. This is later enhanced to operational level \cite{Wiseman}. Many interesting features of the Wigner's Friend scenario have been further analyzed from different perspectives \cite{MZ,Matzkin2020epl,Matzkin2020pra,Elouard2021}. One can also understand the no-go theorems from Brukner \cite{Brukner} and later in \cite{Wiseman}, as witnesses of quantum correlations that are incompatible with AoM.

Another interesting consequence of our work is some interpretations of quantum theory. For all the other interpretations apart from Copenhagenish type interpretations in which there exists operational inconsistency, our results provide an experimental way to refute some of the interpretations of quantum theory. Any violation of \eqref{Tleq1/2}, \eqref{Tqb}, or \eqref{bcP} would rule out universal quantum theory with AoM def-\eqref{AOM} and collapse theories. In Appendix \ref{app:B}, we briefly discuss different interpretations of QT based on the perceptions of AoM and NoM, and whether they lead to consistent predictions.
The task of activation of bipartite correlations shows that non-locality can be activated from local correlations in UQT with NoM but not with AoM. Thus, a resource that was inaccessible in UQT with AoM can be activated in UQT with NoM. 

Further, we provided an operational framework to explore interesting correlations where macroscopic systems like observers or classical registers are treated as physical systems governed by universal physical laws. An important premise in our approach is that any physical system can be described according to quantum theory. We believe that this assumption can be relaxed and the theorems provided in this manuscript can be extended for any general probabilistic theory with certain properties. Our framework can be used to create more exotic scenarios where the distinctions among different interpretations of quantum theory become more relevant. 
Also, it would be interesting to see if the results presented in this manuscript hold true even for any reversible quantum channel instead of unitary transformations applied by Wigner. Also, it might be interesting from a foundational perspective, to find necessary and sufficient conditions that fully characterize quantum correlation with absoluteness of measurement.
\\
\\
\subsection*{Acknowledgements} 
We are thankful
to M. S. Leifer for his
lectures on quantum foundations at the {\it Perimeter Institute
Recorded Seminar Archive}. This work is supported by the Foundation for Polish Science through the First Team project
(No First TEAM/2017-4/31) and the Department of Science \& Technology, India, through the National Post-Doctoral Fellowship (PDF/2020/001682).

\onecolumngrid

\appendix

\section{Proof of Result \ref{thm:P1}}\label{app}


\begin{proof}[Proof of Result \ref{thm:P1}]
Let us say the reversible process $U = \bigoplus_x U^x$ acts on Lab as
\begin{eqnarray}
&U^x\ket{F_0^x}&=\alpha_{0,0}^x\ket{F_0^x}+\alpha_{0,1}^x\ket{F_1^x},\quad \text{and}\nonumber\\
&U^x\ket{F_1^x}&=\alpha_{1,0}^x\ket{F_0^x}+\alpha_{1,1}^x\ket{F_1^x}
\end{eqnarray}
where $\ket{F_i^x}=\ket{\psi_i^x}\ket{f_i^x}$ represent the state \eqref{F} of the respective Labs when Alice observes an outcome $i\in\{0,1\}$ after performing the measurement $A_x$, and
\beq \label{nalpha}
|\alpha^x_{a,0}|^2+|\alpha_{a,1}^x|^2 =  1 \quad \forall x,a = 0,1.
\eeq 
The fact that $\ket{F^x_0},\ket{F^x_1}$ are orthonormal, implies $|\alpha^x_{0,0}\alpha^x_{1,0}| = |\alpha^x_{0,1} \alpha^x_{1,1}|$. Consequently, it follows from \eqref{nalpha} that,
\ben \label{oalpha}
|\alpha^x_{0,0}|^2+|\alpha_{1,0}^x|^2 = |\alpha^x_{0,1}|^2+|\alpha_{1,1}^x|^2 =  1, \quad &\forall x = 0,1.
\een 
On the other hand, using the no-signalling conditions \eqref{Nosigcond}  for any joint probability, we have
\begin{eqnarray}
p(a,b|x,y,U)&=& p(b|B_y)p(a|x,U,y,b)\nonumber\\
&=& p(b|B_y)\left(|\alpha_{a,a}^x|^2p(a|x,\I,y,b)+|\alpha_{a\oplus 1,a}^x|^2p(a\oplus 1|x,\I,y,b)\right) .
\end{eqnarray}
Thus, we can write the joint probability distribution in the presence of super-observer who can access Alice's Lab as,
\begin{eqnarray}
p(a,b|x,y,U)=|\alpha_{a,a}^x|^2p(a,b|x,y,\I)+|\alpha_{a\oplus 1,a}^x|^2p(a\oplus 1,b|x,y,\I) .
\end{eqnarray}
Using the above, the expression of $P_1$ in Eq. \eqref{SP2} can now be simplified as,
\begin{eqnarray}
P_1 = \sum_{a,b,x,y} c_{a,b,x,y} \left(|\alpha_{a,a}^x|^2p(a,b|x,y,\I)+|\alpha_{a\oplus 1,a}^x|^2p(a\oplus 1,b|x,y,\I)\right) .
\end{eqnarray}
Expanding the above quantity, we find

\begin{eqnarray}
P_1 & = & \frac{1}{4} \bigg( |\alpha_{0,0}^0|^2 p(0,0|0,0,\I) + |\alpha_{1,1}^0|^2 p(1,1|0,0,\I) + |\alpha_{1,0}^0|^2 p(1,0|0,0,\I) + |\alpha_{0,1}^0|^2 p(0,1|0,0,\I) \nonumber \\
&& + |\alpha_{0,0}^0|^2 p(0,0|0,1,\I) + |\alpha_{1,1}^0|^2 p(1,1|0,1,\I) + |\alpha_{1,0}^0|^2 p(1,0|0,1,\I)+ |\alpha_{0,1}^0|^2 p(0,1|0,1,\I) \nonumber\\
&& + |\alpha_{0,0}^1|^2 p(0,0|1,0,\I)+ |\alpha_{1,1}^1|^2 p(1,1|1,0,\I) + |\alpha_{1,0}^1|^2 p(1,0|1,0,\I) + |\alpha_{0,1}^1|^2 p(0,1|1,0,\I) \nonumber\\
&& + |\alpha_{0,0}^1|^2 p(0,1|1,1,\I)+ |\alpha_{1,1}^1|^2 p(1,0|1,1,\I) +|\alpha_{1,0}^0|^2 p(1,1|1,1,\I)+|\alpha_{0,1}^1|^2 p(0,0|1,1,\I) \bigg) .
\end{eqnarray}
Since the above quantity is a linear function of the eight variables $\{|\alpha_{a,a}^x|^2, |\alpha_{a,a^\perp}^x|^2 \}$ satisfying \eqref{nalpha} and \eqref{oalpha}, it suffices to consider the extremal values of these variables to obtain its upper bound. It can be checked that there are four extremal values of these variables as follows
\ben 
& (i) \quad & |\alpha^0_{0,0}|^2 = |\alpha^0_{1,1}|^2 = |\alpha^1_{0,0}|^2 = |\alpha^1_{1,1}|^2 = 1, \nonumber \\ 
&& |\alpha^0_{0,1}|^2 = |\alpha^0_{1,0}|^2 = |\alpha^1_{0,1}|^2 = |\alpha^1_{1,0}|^2 = 0 ; \nonumber \\
& (ii) \quad & |\alpha^0_{0,0}|^2 = |\alpha^0_{1,1}|^2 = |\alpha^1_{1,0}|^2 = |\alpha^1_{0,1}|^2 = 1, \nonumber \\ 
&& |\alpha^0_{0,1}|^2 = |\alpha^0_{1,0}|^2 = |\alpha^1_{1,1}|^2 = |\alpha^1_{0,0}|^2 = 0 ; \nonumber \\ 
& (iii) \quad & |\alpha^0_{1,0}|^2 = |\alpha^0_{0,1}|^2 = |\alpha^1_{0,0}|^2 = |\alpha^1_{1,1}|^2 = 1, \nonumber \\ 
&& |\alpha^0_{0,0}|^2 = |\alpha^0_{1,1}|^2 = |\alpha^1_{0,1}|^2 = |\alpha^1_{1,0}|^2 = 0 ; \nonumber \\
& (iv) \quad & |\alpha^0_{1,0}|^2 = |\alpha^0_{0,1}|^2 = |\alpha^1_{1,0}|^2 = |\alpha^1_{0,1}|^2 = 1, \nonumber \\ 
&& |\alpha^0_{0,0}|^2 = |\alpha^0_{1,1}|^2 = |\alpha^1_{0,0}|^2 = |\alpha^1_{1,1}|^2 = 0 .
\een 
These four extremal points correspond to four different CHSH expressions under relabelling of inputs, and accordingly,
\begin{eqnarray}\label{4chsh}
P_1\leqslant \frac{1}{4}\max\{&\sum_{a,b,x,y}& p(a\oplus b=x\cdot y|x,y,\I),\nonumber\\
&\sum_{a,b,x,y}& p(a\oplus b=(x\oplus1)\cdot (y\oplus1)|x,y,\I),\nonumber\\
&\sum_{a,b,x,y}& p(a\oplus b=(x\oplus1)\cdot y|x,y,\I),\nonumber\\
&\sum_{a,b,x,y}& p(a\oplus b=x\cdot (y\oplus1)|x,y,\I)\} .
\end{eqnarray} 
Note that the first CHSH expression is the same as $P_0$, while the second CHSH expression is nothing but $1 - P_0$. The last two CHSH expressions are bounded by $3/2 - P_0$. For instance, the third expression
\ben 
&& \frac{1}{4} \sum_{a,b,x,y} p(a\oplus b = (x\oplus1)\cdot y|x,y,\I) \nonumber \\
&=& \frac{1}{4} \big[ 2 + 2( p(0,0|0,0,\I) + p(1,1|0,0,\I) + p(0,0|1,0,\I)  + p(1,1|1,0,\I) ) \big] - P_0  \nonumber \\
&\leqslant & \frac{3}{2} - P_0  ,
\een 
where we use the fact that $p(0,0|0,0,\I) + p(1,1|0,0,\I)\leqslant 1,$ $p(0,0|1,0,\I)  + p(1,1|1,0,\I) \leqslant 1$. 
Thus, for any quantum strategy, 
\be \label{P1bmax}
P_1 \leqslant \max\{P_0,1-P_0,3/2 - P_0\} = \max\{P_0,3/2 - P_0\}.\ee 
  On the other hand, we notice that each of the four expressions appearing on the right-hand side of \eqref{4chsh} is nothing but the CHSH expression with different re-labeling of the measurement settings $x,y$. Therefore, each of these four expressions is bounded by Tsirelson's bound $1/2(1+1/\sqrt{2})$. This fact, together with \eqref{P1bmax}, we conclude that \eqref{bP1m} holds.    
\end{proof}

\section{Brief discussions on interpretations of quantum theory}\label{app:B}


Here, we briefly discuss different interpretations of universal quantum theory and classify whether they reconcile with `UQT with AoM' or `UQT with NoM'. The overall picture is summarized in Table \ref{table}. \\

\textit{Measurements are inherently classical.} Interpretations of QT like collapse theories \cite{collapse1, collapse2, collapse3, collapse4, collapse5, collapse6}, the ETH approach \cite{ETH} and CSM-ontology \cite{CSM}, to name a few, fall under the category of interpretations in which all physical systems undergo a non-reversible dynamics subject to any measurement process. As a result, in such theories, macroscopic objects, like, measurement devices, follow a different description than microscopic systems. These theories impose that the act of measurement physically changes the state of the system. 
As a consequence, they fall under the category of UQT with AoM, and do not show any violation of \eqref{Tleq1/2}, \eqref{Tqb}, or \eqref{bcP}. \\

\textit{Realist interpretations.} A realist model of UQT assumes the existence of some physical state that completely describes the objective reality of the corresponding physical system and predicts the outcome of all experiments. There are majorly two important classes of realist models. One is where the quantum state is a part of the physical reality, which is commonly referred to as $\psi$-ontic models \cite{Bell66, psi-ontic1, Bohm1}. The other is where the quantum state represents only information about some underlying physical reality, commonly referred to as $\psi$-epistemic models \cite{psi-epis1, psi-epis2, psi-epis3, psi-epis4, psi-epis5}. Any realist description of QT must satisfy some very interesting no-go theorems \cite{Bell, KS, PNC, BOD, UC}. Moreover, there are absolute no-go theorems \cite{LM,Leiferprl,CB,BOD} that constrain $\psi$-epistemic models to a certain extent. 

Interestingly, the original Wigner's Friend thought experiment is in itself a simple no-go theorem against all $\psi$-ontic models applicable to subsystems if the universality of QT holds. In contrast, there exist two kinds of $\psi$-ontic explanations of quantum phenomena that are compatible with the notion of `UQT with NoM'. They are Bohmian mechanics and many-worlds or relative state formalism. According to Bohmian mechanics \cite{Bohm1, Bohm2, Bohm3, Bohm4, Bohm5}, there is a unique description of the whole universe (which the thought-experiment is a part of) that is compatible with the perspective of Wigner. In many-worlds interpretation and relative-state formalism of QT \cite{many-worlds1, many-worlds2, many-worlds3, many-worlds4, many-worlds5}, measurements are not absolute facts in the sense that they do not yield a single outcome. According to these two interpretations, the whole universe evolves via unique reversible dynamics and thus, these two are compatible with NoM.
As a consequence, we would observe violations of \eqref{Tleq1/2}, \eqref{Tqb}, and \eqref{bcP}. \\

\textit{Copenhagenish type interpretations.} Now, we focus on copehangenish type interpretations of QT. For our analysis, we refer to the lectures by Leifer \cite{leifer}. All such interpretations of QT are based on the following principles. First, the universality of quantum predictions imposes that every system is describable using QT. Secondly, the anti-realism of physical theory is that there is no objective description of reality. And thirdly, every measurement performed by an observer yields a single outcome to that observer. These principles imply the existence of a split known as the Heisenberg cut that separates the observer from the physical system which is being described by the observer. The placement of the Heisenberg cut is quite significant as it represents the boundary up to which an observer can describe physical systems using QT and must invoke classical theory after this split. 
In the context of the WFS scenario, we note that if the split is placed before the boundary of the isolated Lab, that is, if the split is placed on the left side of the dotted line in Fig. \ref{fig}, then the description of the isolated Lab is unambiguous and compatible with AoM \eqref{AOM}. On the contrary, if the split is placed after the boundary of the isolated Lab, that is, if the split is placed anywhere between the dotted line and the respective observer (Wigner or Student) in Fig. \ref{fig}, then the tension between the two dynamics of UQT persists, and as a result, we will arrive at operational inconsistency. 

The Copenhagenish type interpretations of QT are broadly divided into two categories: Objective interpretations, which include the original Copenhagen interpretation by Bohr \cite{Bohr}, Quantum pragmatism \cite{Prag}, and information interpretation \cite{Infointer1, Infointer2}; and Perspectival interpretations, which include Qbism \cite{Qbism1, Qbism3}, Consistent histories \cite{ch}, and Relational quantum mechanics \cite{Rovelli}. The difference between these two types lies in the fact that observing a measurement outcome by an observer (say, Friend) is an objective fact in the former type, while it is considered a subjective fact for that observer in the latter type. In the Perspectival approach, the placement of the split is decided by the observer. Therefore, operational consistency does not hold; however, this approach does not necessitate operational consistency. Moreover, depending on where the split is placed by an observer, violations of \eqref{Tleq1/2}, \eqref{Tqb}, and \eqref{bcP} prevail. On the other hand, it is not fully clear whether the placement of the split depends on the observer or not in Objective Copenhagenish interpretations, and thus, there are ambiguities therein.


\begin{table}[h!]
    \centering
\begin{tabular}{ |c | c | c | c | c|}
\hline
 Interpretations & UQT with AoM & UQT with NoM & \begin{tabular}{@{}c@{}} OCO \end{tabular} & 
 \begin{tabular}{@{}c@{}}
    Violation of conditions \\
 \eqref{Tleq1/2}, \eqref{Tqb},
    \eqref{bcP}
    \end{tabular} \\
\hline 
\hline
\begin{tabular}{@{}c@{}}
     Spontaneous collapse theories \cite{collapse1, collapse2, collapse3, collapse4, collapse5, collapse6},  \\
      ETH approach \cite{ETH}, CSM \cite{CSM}
\end{tabular}
  & $\checkmark$ &  \text{\sffamily X}  & $\checkmark$ & \text{\sffamily X}\\
 \hline 
 \begin{tabular}{@{}c@{}}
  Bohmian mechanics \cite{Bohm1, Bohm2, Bohm3, Bohm4, Bohm5},   \\ 
  Many-world, Relative-state formalism \cite{many-worlds1, many-worlds2, many-worlds3, many-worlds4, many-worlds5}  
 \end{tabular}
 & \text{\sffamily X} & $\checkmark$ & $\checkmark$ & $\checkmark$\\
\hline
Perspectival Copenhagenish \cite{Qbism1, Qbism3,ch, Rovelli} & \text{\sffamily X} & \text{\sffamily X} & \text{\sffamily X} & $\checkmark$ \\

  \hline 
Objective Copenhagenish \cite{Prag, Infointer1, Infointer2} & \text{\sffamily X} &  ? & ? & ?\\
  \hline
\end{tabular}
\caption{Classification of interpretations of quantum theory based on - $(i)$ whether they are compatible with absoluteness of measurement \eqref{AOM} or non-absoluteness of measurement \eqref{AOE}, $(ii)$ whether they are OCO according to def-\eqref{def:oc}, and $(iii)$ whether they provide violations of \eqref{Tleq1/2}, \eqref{Tqb}, and \eqref{bcP}. For the Objective Copenhagenish type of interpretation, there are ambiguities with respect to such classification.} 
\label{table}
\end{table}


\end{document}